\documentclass{article}
\usepackage{amsmath, amssymb, amsthm}
\usepackage{mathtools}
\usepackage{booktabs}
\usepackage{enumitem}
\usepackage{float}
\usepackage{microtype}
\usepackage{graphicx}   
\usepackage{xcolor}     
\usepackage{tikz}       
\usetikzlibrary{arrows.meta,positioning,shapes.geometric,calc,fit} 
\usepackage{rotating}   
\usepackage{natbib}
\usepackage[hidelinks]{hyperref}
\usepackage[nameinlink,noabbrev]{cleveref}
\usepackage[toc,page]{appendix}

\theoremstyle{plain}
\newtheorem{proposition}{Proposition}[section]
\newtheorem{lemma}[proposition]{Lemma}
\newtheorem{assumption}[proposition]{Assumption}
\newtheorem{problem}[proposition]{Problem}
\theoremstyle{remark}
\newtheorem*{remark}{Remark}

\numberwithin{equation}{section}

\newcommand{\E}{\mathbb{E}}

\newcommand{\PP}{\mathbb{P}}
\newcommand{\dd}{\mathrm{d}}
\newcommand{\barF}{\bar F}
\newcommand{\haz}{h}
\newcommand{\R}{\mathbb{R}}

\setlist[itemize]{leftmargin=*, itemsep=2pt, topsep=2pt}
\setlist[enumerate]{leftmargin=*, itemsep=2pt, topsep=2pt}

\title{Optimal Transfer Mechanism for Municipal Soft Budget Constraints in Newfoundland}
\author{Xinli (Russel) Guo\thanks{Email: \texttt{xinli.guo@etu.univ-paris1.fr}. I wish to thank Professor Marcus Pivato and Professor Hélène Huber for their
valuable suggestions and encouragement. I also acknowledge the helpful discussions with peers at Université Paris 1 Panthéon-Sorbonne, which contributed to
improving the clarity of the analysis. All remaining errors are my own.}}
\date{Auguest 2025}

\begin{document}

\maketitle

\begin{abstract}
Newfoundland and Labrador's municipalities face severe soft budget pressures due to narrow tax bases, high fixed service costs, and volatile resource revenues. We develop a Stackelberg style mechanism design model in which the province commits at $t=0$ to an ex ante grant schedule and an ex post bailout rule. Municipalities privately observe their fiscal need type, choose effort, investment, and debt, and may receive bailouts when deficits exceed a statutory threshold. Under convexity and single crossing, the problem reduces to one dimensional screening and \emph{admits} a tractable transfer mechanism with quadratic bailout costs and a statutory cap. The optimal \emph{ex ante} rule is \emph{threshold--cap}; under discretionary rescue at $t=2$, it becomes \emph{threshold--linear--cap}. A knife-edge inequality yields a self-consistent no bailout regime, and an explicit discount factor threshold renders hard budgets dynamically credible. We emphasize a class of monotone \emph{threshold} signal rules; under this class, grant crowd out is null almost everywhere, which justifies the constant grant weight used in closed form expressions. The closed form characterization provides a policy template that maps to Newfoundland's institutions and clarifies the micro-data required for future calibration.
\end{abstract}


\medskip
\newpage

\section{Introduction}\label{sec:intro}
Local public finances in Newfoundland and Labrador (NL) have long faced a
three--way squeeze: a limited own--source tax base, large fixed costs of delivering
services to sparsely populated areas, and the province’s chronic macro--fiscal
stress.\footnote{NL's per-capita debt-service burden has ranked first among
Canadian provinces since at least 2020; see \citet{nl_finance2020}.}
Few Canadian provinces illustrate the political economy of
\emph{soft budget constraints} (SBCs) so vividly as
NL. 

SBC, described by
\citet{kornai1986}, arise once an upper-tier government develops a bailout record and
lower-tier entities rationally expand spending or borrowing in expectation of
future relief. On the revenue side, sparsely populated outports yield a narrow
own-source tax base; on the cost side, geography and winter logistics impose
one of the highest per-capita service bills in the country.
When commodity downturns hit---most recently in 2015--2016---the
province repeatedly ``stepped in'': it assumed municipal debts \emph{in toto},
stretched repayment schedules, and negotiated ad hoc funding packages with
Ottawa.\footnote{See the 25 May 2016 press release,
\citet{GovNFL_2016}.}
By normalizing such rescues NL has created precisely the expectation of
future relief that \citet{kornai1986} warned about.

Although a rich literature analyzes SBCs in transition economies and US/EU
federal systems, two limitations stand out:

\begin{enumerate}[label=(\roman*),leftmargin=*]
\item \textbf{Timing.}  
      Most formal models cast upper and lower tiers as simultaneous Nash
      players; real-world transfers are decided \emph{sequentially}.
\item \textbf{Granularity.}  
      Empirical work on Canada concentrates on provincial--federal
      equalization; the municipal layer---where soft budgets can first bite---has
      received scant theoretical attention
      \citep{bird_2012,sancton_govcfs_2014,boothe_2015}.
\end{enumerate}

In this paper, we build a \emph{Stackelberg‐style screening model} in which the
province first commits to a two-part menu
$\bigl(T(\hat\theta),b(\hat\theta)\bigr)$; a municipality then
privately observes its \emph{fiscal need} type~$\theta$ and reports $\hat\theta$; finally the
province may---\emph{ex post}---augment the transfer with a bailout schedule $\beta(\cdot)$ that maps the noisy gap signal $\hat G$ to a payment.
We adopt the \textbf{implementable payout convention}
\[
  \text{Realized payout}\quad p(\hat G,\hat\theta)
  \;=\; \mathbf 1\{\hat G>0\}\cdot \min\bigl\{\,\beta(\hat G),\; b(\hat\theta),\; \hat G\,\bigr\},
\]
so that the \emph{signal-based rule} $\beta$ is hard-capped by the \emph{type-based cap} $b$ and by the observed signal. This aligns the model with administrative practice, avoids paying above the observed gap, and acknowledges that mis-payment risk stems only from signal noise.

The model is theoretical, yet four closed-form results offer
a tractable benchmark for practitioners and future empirical work in NL:

\begin{enumerate}[label=(\alph*),leftmargin=*]
\item \textbf{Dimensionality reduction.}
      The three-stage game collapses to one-dimensional adverse selection with two scalar instruments $(T,b)$; see Proposition~\ref{prop:reduction}.

\item \textbf{Optimal transfer rule and knife-edge.}
      With linear--quadratic provincial costs and a statutory cap, the IC--IR--LL optimum is \emph{threshold--cap}; under a discretionary $t{=}2$ rescue (no commitment), the realized rule is \emph{threshold--linear--cap} (Appendix~\ref{subsec:disc}).
      A self-consistent no-bailout regime obtains whenever
      \[
      \boxed{\,\alpha\omega_T \;\ge\; \gamma\omega_b \cdot \sup_{\theta} \haz(\theta)\,},\qquad
      \haz(\theta)=\tfrac{f(\theta)}{\barF(\theta)},
      \]
      see Proposition~\ref{lem:nobailout} and \citet{BarlowProschan1975,ShakedShanthikumar2007}.

\item \textbf{Policy template for NL.}
      The resulting two-parameter grant formula $(\theta^{\min},\theta^{\dagger})$ maps to NL's Municipal Operating Grant and identifies the micro-data needed to calibrate $(\gamma,\kappa,\lambda_T,\alpha,\omega_b)$; see Proposition~\ref{prop:opt_cap}.
\end{enumerate}

\noindent\textit{Roadmap.} Section~\ref{sec:model} develops the Stackelberg model and derives the
reduced form.  
Section~\ref{sec:Tstar} solves for the optimal transfer schedule and proves its
second-best efficiency, introducing the thresholds $\theta^{\min}$ and $\theta^{\dagger}$ (defined in~\eqref{eq:theta-dagger}).  
Section~\ref{sec:policy} draws policy lessons and comparative statics for NL.  
Section~\ref{sec:conclusion} concludes and outlines empirical extensions.

\section{Literature Review}\label{sec:lit}

The paper intersects three strands of work:  
(i)~soft budget constraints (SBC) in multi-tier public finance,  
(ii)~mechanism design with Stackelberg leadership, and  
(iii)~Canadian municipal-finance empirics.

\subsection{Soft Budget Constraints}

The SBC idea begins with \citet{kornai1986}.  Early formalizations
\citep{kornai_maskin_roland_2003} show that ex post efficient bailouts
undermine ex ante effort and borrowing; see also
\citet{dewatripont_maskin_1995} for a dynamic commitment model.
Applications to intergovernmental finance include
\citet{weingast_1995} on U.S.\ states,
\citet{bordignon_2001} on European stability pacts, and the survey
by \citet{goodspeed_2016}.
Recent papers ask \emph{when} an upper tier can credibly refuse rescues:
\citet{AmadorEtAl2021AER} derive fiscal limits under limited commitment,
\citet{PavanSegal2023JET} study repeated screening, and
\citet{AcemogluJackson2024RESTUD} analyze relational contracts with hidden
actions.  Yet these models stop short of a closed-form, policy-ready transfer
rule.  We fill that gap by producing an implementable
\emph{threshold--cap} schedule under commitment (and \emph{threshold--linear--cap} under $t{=}2$ discretion) and a single discount-factor test for
self-enforcing hard budgets.

\subsection{Mechanism Design and Stackelberg Leadership}

Incentive-compatible grant design dates back to
\citet{bordignon_montolio_picconi_2003}, who use a single matching grant under
simultaneous moves.  \citet{toma_2013} introduces leader--follower timing but
assumes full information.  \citet{chen_silverman_2019} obtain threshold
payments in a one-shot health-care model with asymmetric costs, yet ignore
ex post instruments and dynamic credibility.  
Our contribution is twofold:  
(i)~a \textit{two-instrument} Stackelberg screen that nests bailout and
no bailout regimes in the marginal-cost inequality
$\alpha\omega_T\gtreqless\gamma\omega_b$ (with the exact knife-edge refined by the hazard bound~\eqref{lem:nobailout});  
(ii)~a closed-form critical discount factor that pins down dynamic
self-enforcement.

\subsection{Canadian Municipal Finance}

Empirical work on Canadian municipalities examines fiscal capacity and service
costs \citep{bird_2012,sancton_govcfs_2014} or tax-base sharing
\citep{dahlby_ferede_2021}, and analyzes borrowing limits
\citep{found_tompson_2020}.  Evidence on municipal-level SBC is scarce; most
studies focus on federal--provincial equalization
\citep{boothe_2015}.  A notable exception is
\citet{BraccoDoyle2024JPubE}, who exploit BC’s 2004 debt-limit reform, and the
cross-country benchmark in \citet{Rodden2006}.
None, however, links provincial bailout policy to a mechanism-design
benchmark.  Our theory provides that benchmark and outlines the data needed
for future identification once municipal micro-panels for NL---or other provinces---become available.

\section{Provincial--Local Fiscal Relations in Newfoundland and Labrador}
\label{sec:background}

Although the model is self-contained, its parameters map neatly onto four
institutional frictions that characterize NL.
Documenting those frictions clarifies both the choice of primitives and the
comparative-statics exercises that follow.

\subsection{A Heterogeneous Local Landscape}

In NL, more than 260 local units fall into two legal categories.  Incorporated
municipalities possess broad tax powers---chiefly the property tax---whereas
unincorporated Local Service Districts (LSDs) finance specific services
through earmarked levies.  LSDs are typically small, sparsely populated, and
administratively thin; the province therefore faces persistent political
pressure to guarantee minimum service levels even when the local tax base is
weak.  We model this dispersion of own-source capacity as private information:
each jurisdiction’s \emph{fiscal need type} \(\theta\) indexes the severity of its funding gap (higher $=$ weaker capacity / larger need).

\subsection{Transfer Instruments and Model Notation}

Three provincial channels matter and correspond one-for-one to our model
variables:

\begin{description}[leftmargin=*]
\item[Municipal Operating Grant (MOG).]  
      An \emph{unconditional} operating transfer that forms the baseline
      grant~\(T(\theta)\).

\item[Canada Community--Building Fund (CCBF).]  
      Formula-based or cost-shared capital money.  
      Because these flows are largely exogenous to short-run fiscal gaps we
      treat them as a constant background term $g$ and abstract from them in the
      formal screening problem.

\item[Special Assistance.]  
      Ad-hoc subsidies or debt relief for distressed communities; this is the
      discretionary bailout instrument~\(b(\theta)\).
\end{description}

\subsection{Borrowing Oversight and Commitment Frictions}

Long-term municipal borrowing requires ministerial approval and often carries
an explicit provincial guarantee.  Once a community approaches default the
province internalizes the externality of municipal bankruptcy and almost
always chooses rescue over liquidation.  Anticipating that bias, low-capacity
jurisdictions rationally relax ex ante effort---exactly the soft budget channel
formalized in our Stackelberg mechanism.

\subsection{Stylized Fiscal Facts and Parameter Guidance}

\begin{enumerate}[label=(\alph*),leftmargin=1.2em]
\item \textbf{High provincial debt service.}  
      NL's per-capita debt-service burden tops the Canadian league table,
      implying a high marginal opportunity cost~\(\gamma\) of each grant
      dollar.

\item \textbf{Transfer-dependent small jurisdictions.}  
      For many LSDs, unconditional grants finance over half of current
      spending, whereas bailouts are politically constrained and arrive late.
      Thus \(\omega_T>\omega_b\).

\item \textbf{Administrative capacity gaps.}  
      Uneven record-keeping and audit lags create the information asymmetry
      that justifies the single-crossing (or virtual monotonicity) used in the model.

\item \textbf{Prohibitively high political cost of very large rescues.}  
      Public debate in NL treats bailouts above a ``headline'' threshold as
      politically costly, consistent with the convex cost term
      \(\tfrac{\kappa}{2}b^{2}\) in the provincial objective.
\end{enumerate}

\medskip
\noindent\textit{From institutions to the model.}  
Sections~\ref{sec:model}--\ref{sec:Tstar} fold the institutional environment into a
two-instrument screen $(T,b)$: an ex ante operating grant and an ex post bailout.
All primitives that do not affect screening incentives directly (e.g.\ CCBF capital transfers $g$) are absorbed into fixed terms.

\section{Model}\label{sec:model}

We study the interaction between a single provincial government ~$P$ and a
continuum of local jurisdictions $i\in\mathcal I\subset[0,1]$ representing
municipalities or Local Service Districts (LSDs). 

\paragraph{Type convention.}
Throughout, we define the private type $\theta\in[\underline\theta,\bar\theta]$ as a \emph{fiscal-need index}:
a higher $\theta$ corresponds to a weaker local tax base / higher per-unit service cost,
hence a larger underlying funding gap. This convention is used consistently in the theory and background sections.

\subsection{Technologies and Preferences}

\begin{description}[leftmargin=0cm, labelsep=0.5cm, style=nextline]
  \item\textbf{Basic services} Each local jurisdiction \( i \) delivers a bundle of essential public services
  \( q = (\text{water}, \text{sewer}, \text{roads}, \ldots) \).
  The monetary cost is modeled by \( C(q,\theta) \), where \( \theta \)
  denotes the jurisdiction’s fiscal need.\\
  \emph{Higher} \( \theta \) (narrow tax base, difficult geography)
  \(\Rightarrow\) higher marginal cost.

  \item\textbf{Heterogeneous fiscal need} Fiscal need is heterogeneous across municipalities. We treat \( \theta \) as a draw
  from a continuous distribution \( f(\theta) \) with support
  \( [\underline{\theta},\bar{\theta}] \), so both high- and low-need jurisdictions
  are present in the province.

  \item\textbf{Local effort} Jurisdictional effort \( e \ge 0 \) --- tax enforcement, fee collection, grant writing ---
  generates own-source revenue \( R(e,\theta) \) with $R'_e>0$ and $R''_{ee}<0$.

  \item\textbf{Disutility of effort} Effort imposes a linear utility cost 
  \begin{align*}
  \text{Disutility}=-\phi e, \qquad \phi>0,
\end{align*}
on residents or officials. The linear form keeps derivations tractable while capturing the political and administrative
  cost of higher taxation.

\item\textbf{Service and investment utility.}
Delivering services and undertaking investment generate direct benefits.
Let $B,\Gamma:\mathbb{R}_+\to\mathbb{R}$ be $C^2$ and normalized by $B(0)=\Gamma(0)=0$.
We assume:
\begin{itemize}[label={},leftmargin=0pt,itemsep=0.2em]
\item \textbf{Monotonicity.} $B'(q)>0$ and $\Gamma'(I)>0$ for all $q,I\ge0$.
\item \textbf{Concavity.} $B''(q)\le0$ and $\Gamma''(I)\le0$ for all $q,I\ge0$.
\item \textbf{Diminishing marginal returns.} $\lim_{q\to\infty}B'(q)=0$ and $\lim_{I\to\infty}\Gamma'(I)=0$.
\item \textbf{Bounded marginal utilities.} $\sup_{q\ge0}B'(q)<\infty$ and $\sup_{I\ge0}\Gamma'(I)<\infty$.
\end{itemize}
These conditions ensure well-behaved choices and allow $B(q)+\Gamma(I)$ to enter the objective via bounded terms.
\end{description}

\paragraph{Transfers}
The Province can effect intergovernmental transfers in three distinct ways:
\begin{enumerate}[label=(\roman*)]
\item \textbf{Unconditional grant} $\tau$, set \emph{ex ante} before local effort choices;
\item \textbf{Matching transfer} at provincial share $s$ on capital outlays $I$;
\item \textbf{Ex--post bailout} $b\ge 0$, extended only if a realized gap remains after
fiscal shocks $\varepsilon$ are realized.
\end{enumerate}
Before local choices, the Province commits to three \emph{transfer instruments}.  Table~\ref{tab:transfers_table} summarizes timing and incentives.

Given these transfers, the \textbf{one-period cash-flow constraint} at $t=1$ (before any payout) is
\begin{align}
  G &= \bigl[C\bigl(q,\theta\bigr)+(1-s)I+rD\bigr] 
        -\bigl[R(e,\theta)+\tau+g+sI+D\bigr] + \varepsilon \label{eq:gap_pre}
\end{align}
\footnote{Notation reminder: in the reduced-form mechanism we write $T(\theta)$ for the unconditional operating grant that corresponds to the empirical instrument $\tau$; and $g$ captures largely exogenous capital transfers (e.g.\ CCBF/Gas-Tax) that are treated as constants in screening.}

If $G>0$ a funding gap exists.  Province observes a noisy signal $\hat{G}=G+\eta$ and pays the \emph{realized payout}
\[
  p(\hat G,\hat\theta)
  \;=\; \mathbf 1\{\hat G>0\}\cdot\min\{\beta(\hat G),b(\hat\theta),\hat G\}
  \;\in\;[0,\hat G],
\]
at $t=2$ under commitment to $(\beta,b)$.

\paragraph{Softness metric}
We define the \emph{softness index} as the ex--ante rescue probability
\begin{equation*}
  \pi \;=\; \PP\bigl[p(\hat{G},\hat\theta)\ge G\bigr]\in[0,1],
\end{equation*}
namely the \emph{probability that a realized funding gap will be fully covered}:
$\pi=0$ is a hard budget constraint; $\pi=1$ is full insurance.

\paragraph{Aggregate softness.}
Throughout Sections~\ref{sub:effort}, \ref{subsec:reduction}, and \ref{sec:stage_game} we treat $\pi$ as the \emph{aggregate}
ex ante rescue probability (integrated over types). It is a descriptive index; in the baseline with threshold $\beta$ it does not affect the grant weight derived below.

\subsection{Local effort $e$: incentives and marginal condition}\label{sub:effort}
Effort $e\ge0$ (tax enforcement, fee collection, administrative intensity) generates own--source revenue
$R(e,\theta)$ with $R'_e>0$ and $R''_{ee}<0$.  Effort is personally/politically costly as $-\phi e$.


\begin{assumption}[Primitives]\label{ass:primitives}
Types $\theta$ lie in $[\underline{\theta},\bar{\theta}]$ with density $f>0$.
The local cost and revenue functions satisfy, for all $\theta$,
\[
C'_q(q,\theta)>0,\qquad C''_{qq}(q,\theta)\ge0,\qquad
R'_e(e,\theta)>0,\qquad R''_{ee}(e,\theta)<0.
\]
\end{assumption}

\begin{assumption}[Observation \& rule regularity]\label{ass:regularity}
The signal rule $\beta$ is nondecreasing and a.e.\ differentiable with slope in $[0,1)$; the audit noise $\eta$ has a continuous density $f_\eta$ with bounded tails; and $(e,\theta)\mapsto R'_e(e,\theta)$ is continuous. These ensure interchange of expectation and differentiation and validate the marginal-probability formulas below.
\end{assumption}

\begin{assumption}[Signal MLRP]\label{ass:mlrp}
For $\theta_2>\theta_1$, the family of signals $\{\hat G\mid\theta\}$ satisfies the monotone-likelihood-ratio property (MLRP), so that for any nondecreasing $\varphi$, $\E[\varphi(\hat G)\mid\theta_2]\ge \E[\varphi(\hat G)\mid\theta_1]$. Since $\beta$ is nondecreasing, the cap-slope objects defined below are nondecreasing in $\theta$.
\end{assumption}

\begin{assumption}[Restriction to threshold signal rules]\label{ass:beta-threshold}
We restrict attention to signal-based payout rules $\beta$ that are nondecreasing and piecewise constant in $\hat G$ (threshold rules). This class is consistent with administrative practice and eliminates effort--report interactions almost everywhere.
\end{assumption}

\paragraph{Observation frictions and default.}
Default occurs iff $p(\hat G,\hat\theta)<G$. Because
of information frictions, municipalities rationally expect a positive rescue probability $\pi>0$. Anticipating the chance of a bailout, they optimally reduce tax effort $e$ and rely more on debt $D$.\footnote{For axiomatic treatments of decision under noisy or imperfect perception, see \citet{PivatoVergopoulos2020JME}, which provides a clean way to model observation constraints consistent with our audit-noise setup.}

\begin{assumption}[Effort independence via threshold $\beta$]\label{ass:report-invariance}
Under Assumptions~\ref{ass:primitives}--\ref{ass:beta-threshold}, the first-order condition for effort on the cap-slack branch contains no term depending on the report $\hat\theta$; $e^{\ast}(\theta)$ is report-independent up to boundary-density terms on the cap-binding set.
\end{assumption}

\noindent\emph{Technical detail.} See Lemma~\ref{lem:e-indep} in Appendix~\ref{app:proofs}.

\subsection{Annual timeline: three decision stages}\label{subsec:timeline}

We normalize one fiscal year to $t\in\{0,1,2\}$; the sequence repeats every year.

\paragraph{Stage $t=0$ (policy commitment).}
The province announces a vector
$\Pi=(\tau,s,\bar D,\beta,b)$ where $\beta:\R_+\to\R_+$ is the signal-based payout rule implemented at $t=2$ and $b(\cdot)$ is a type-based cap.

\paragraph{Stage $t=1$ (local choices and gap realization).}
After observing its private type $\theta\sim f$, the municipality selects effort $e\ge0$, capital $I\ge0$, and debt $0\le D\le\bar D$. A mean-zero fiscal shock $\varepsilon$ is then realized and the pre-payout gap is $G$ from \eqref{eq:gap_pre}.

\paragraph{Stage $t=2$ (signal, payout, default test).}
The Province observes $\hat G=G+\eta$ and pays $p(\hat G,\hat\theta)=\mathbf 1\{\hat G>0\}\min\{\beta(\hat G),\,b(\hat\theta),\,\hat G\}$. Default occurs iff $p(\hat G,\hat\theta)<G$.

\subsection{Local optimization at $t=1$}

Given policy $\Pi$, the municipality solves
\begin{equation}\label{eq:local_prog}
\begin{aligned}
\max_{e,q,I,D}\; & 
  \E_{\varepsilon,\eta}\!\Bigl[
      B(q)+\Gamma(I)+R(e,\theta)-\phi\,e
      -\varphi\,\mathbf1\!\{\,p(\hat G,\hat\theta)<G\,\}
      +\omega_b\,p(\hat G,\hat\theta)
  \Bigr] \\[ -2pt]
\text{s.t.}\quad
&G = C(q,\theta)+(1-s)I+rD
      -\bigl[R(e,\theta)+\tau+g+sI+D\bigr]
      +\varepsilon, \\[4pt]
&0 \le D \le \bar D,\qquad e,q,I \ge 0 .
\end{aligned}
\end{equation}
Here $\omega_b>0$ is treated as a constant marginal-utility weight that
local decision-makers attach to each dollar of bailout; see
Remark~\ref{rem:omega-const}.

\paragraph{Effort incentives via the default and marginal-rescue channels.}
We report the marginal default formula for the \emph{signal branch} (cap slack), which is the relevant case almost everywhere under threshold/linear $\beta$; the cap only binds in an upper region where the marginal-rescue channel shuts down.

\noindent\emph{Technical detail.} See Lemma~\ref{lem:margprob} in Appendix~\ref{app:proofs}.

\noindent
Since $\partial \E[p(\hat G,\hat\theta)]/\partial e=-R'_e(e,\theta)\,\E[\beta'(\hat G)\,\mathbf1\{\beta(\hat G)<b(\hat\theta)\}]$, the interior FOC on the signal branch is
\begin{equation}\label{eq:FOC_e_correct}
  R'_e\!\bigl(e^{\ast}(\theta);\theta\bigr)\,
  \Bigl\{
    1
    +\varphi\,\E\!\left[f_\eta\!\bigl(\hat{G}-\beta(\hat{G})\bigr)\bigl(1-\beta'(\hat G)\bigr)\right]
    -\omega_b\,\E\!\bigl[\beta'(\hat G)\,\mathbf1\{\beta(\hat G)<b(\hat\theta)\}\bigr]
  \Bigr\}
  = \phi .
\end{equation}
For threshold (piecewise constant) $\beta$, $\beta'(\hat G)=0$ a.e., hence the last term vanishes.

\paragraph{Comparative statics.}
From \eqref{eq:FOC_e_correct}, $\partial e^{\ast}/\partial \varphi>0$ and $\partial e^{\ast}/\partial f_\eta>0$.
The sign of $\partial e^{\ast}/\partial \omega_b$ is $\le 0$ through the $-\omega_b\,\E[\beta'(\hat G)\mathbf1\{\beta<b\}]$ term; under threshold rules, $\partial e^{\ast}/\partial \omega_b=0$.


\subsection{Reduction to One--Dimensional Screening}
\label{subsec:reduction}

The three--stage environment features effort $e$, capital $I$, debt $D$, service $q$, and shocks $(\varepsilon,\eta)$. We now show
that---under the convexity/monotonicity conditions already stated---these
objects can be optimized out, leaving a one--dimensional adverse--selection
problem with quasi--linear utility in $(T,b)$, where $b$ is a \emph{cap parameter}.

\begin{assumption}[Rare cap binding]\label{ass:cap-slack}
There exists $\varepsilon<1$ such that
$\PP_\theta[\beta(\hat G)\ge b(\theta)]\le\varepsilon$ for every feasible
cap profile $b(\cdot)$ in the optimal mechanism.
This ensures the cap binds only on an upper tail with measure $\le\varepsilon$
and allows us to treat $\tilde b(\hat\theta;\theta)$ and $e^\star(\theta)$ as
\emph{$O(\varepsilon)$-close} to the cap-slack expressions used in
\eqref{eq:UL-reduced}.
\end{assumption}

\noindent\emph{Clarification.} Assumption~\ref{ass:cap-slack} refers to the probability
(over the audit noise $\eta$) that a given type's realization lies on the cap-binding
branch; it does not exclude the existence of a type region $[\theta^\dagger,\bar\theta]$
where the cap $b(\theta)=\bar b$ binds ex ante.

\paragraph{Step~1. Local optimization.}
Fix $(\tau,s,\bar D,\beta)$ and a true type $\theta$. Minimizing the pre-bailout resource block delivers $(q^{\ast},I^{\ast},D^{\ast})$ and reduced cost $C_0(\theta)$, independent of the report.

\paragraph{Step~2. Effort choice.}
Effort $e^\ast(\theta)$ is pinned down by \eqref{eq:FOC_e_correct} (cap slack almost everywhere).

\paragraph{Step~3. Quasi-linear reduced form with payout cap.}
Define the \emph{expected payout under cap}, conditional on type,
\[
  \tilde b(\hat\theta;\theta)
   = \E\bigl[\min\{\beta(\hat G),b(\hat\theta)\}\,\big|\,\theta\bigr].
\]
Then the interim utility from reporting $\hat\theta$ can be written
\begin{equation}
  U_L(\hat\theta,\theta)=\lambda_T(\theta)\,T(\hat\theta)
    + \omega_b\,\tilde b(\hat\theta;\theta)+K(\theta),
  \label{eq:UL-reduced}
\end{equation}
with a grant weight
\[
  \lambda_T(\theta)\;=\;\omega_T\; -\; \omega_b\,\E\Bigl[\frac{\partial p}{\partial T}\,\Big|\,\theta\Bigr].
\]
Under Assumption~\ref{ass:beta-threshold} and continuous noise, $\beta'(\hat G)=0$ almost everywhere, so $\lambda_T(\theta)=\omega_T$.

\begin{lemma}[Grant crowd-out factor]\label{lem:lambdaT}
With $\hat G=G+\eta$ and $G$ decreasing in $T$ one-for-one, the marginal effect of $T$ on the realized payout is
\[
  \frac{\partial}{\partial T}\E[p(\hat G,\hat\theta)\mid\theta]
  = -\,\E\!\left[\beta'(\hat G)\,\mathbf 1\{\beta(\hat G)<b(\hat\theta)\}\,\middle|\,\theta\right]
  + \text{boundary terms}.
\]
If $\beta$ is threshold (piecewise constant), $\beta'(\hat G)=0$ a.e. and the boundary terms vanish under continuous noise, hence $\lambda_T(\theta)=\omega_T$.
If $\beta$ has an interior linear branch with slope $m\in(0,1)$ on the cap-slack set, then $\lambda_T(\theta)=\omega_T-\omega_b\,m\cdot \PP_\theta[\text{cap slack and } \hat G \text{ in linear range}]$.
\end{lemma}

\begin{proposition}[Reduction]\label{prop:reduction}
Under Assumptions~\ref{ass:primitives}, \ref{ass:regularity}, \ref{ass:mlrp} and~\ref{ass:beta-threshold}, the original
moral-hazard problem is equivalent to a direct mechanism in which
municipal interim utility is quasi-linear in the \emph{cap parameter} $b$ through $\tilde b(\hat\theta;\theta)=\E[\min\{\beta(\hat G),b(\hat\theta)\}\mid\theta]$ as in \eqref{eq:UL-reduced}.
\end{proposition}

\begin{remark}[Constant marginal utility of bailouts]\label{rem:omega-const}
Lemma~\ref{lem:lambdaT} implies that, under threshold $\beta$, the grant weight equals $\omega_T$. If instead a discretionary linear segment applies (Appendix~\ref{subsec:disc}), the weight falls below $\omega_T$ proportionally to the slope and the probability of being on the linear branch.
\end{remark}

\begin{remark}[Effect of default-loss term and approximation accuracy]\label{rem:defaultloss}
Because the default indicator only flips on the cap-binding set, which has probability at most $\varepsilon$ by Assumption~\ref{ass:cap-slack}, the marginal effect of a report on the $-\varphi\,\mathbf1\{p<G\}$ term is $O(\varepsilon)$ and can be absorbed into $K(\theta)$. On the cap-slack region, the common tail probability cancels in the FOC just as in Lemma~\ref{lem:capcalc}. This formalises the reduction in \eqref{eq:UL-reduced} up to an $O(\varepsilon)$ error.
\end{remark}

\subsection{Single--Period Mechanism Design}
\label{sec:stage_game}
We henceforth work with the reduced form \eqref{eq:UL-reduced}, and with provincial cost taken in \emph{expectation} over the realized payout.

\paragraph{Provincial cost.}
For a cap profile $b(\theta)$ the expected per-type cost is
\[
  \E\Bigl[ \alpha\,\min\{\beta(\hat G),b(\theta)\} + \frac{\kappa}{2}\,\min\{\beta(\hat G),b(\theta)\}^{2} \Bigm|\,\theta\Bigr].
\]
The province minimizes the population expectation of this cost plus $\gamma T(\theta)$, subject to IC, IR and limited liability.

\paragraph{Envelope and weights.}
With $V(\theta)=U_L(\theta,\theta)$ and $V(\underline\theta)=\underline U$,
\begin{equation}\label{eq:envelope-correct}
V'(\theta)=\lambda_T'(\theta)\,T(\theta)+\omega_b\,\partial_\theta\tilde b(\theta;\theta)+K'(\theta).
\end{equation}

\begin{remark}[Baseline]\label{rem:lambdaT-const}
In the baseline with threshold $\beta$ and constant $\omega_b$, we have $\lambda_T(\theta)\equiv\omega_T$, so \eqref{eq:envelope-correct} reduces to $V'(\theta)=\omega_b\,\partial_\theta\tilde b(\theta;\theta)+K'(\theta)$.
\end{remark}

\begin{assumption}[Single crossing in the allocation index]\label{ass:sc}
Define the allocation index for report $\hat\theta$ at true type $\theta$ by
\[
  x(\hat\theta;\theta)\;=\;\lambda_T(\theta)\,T(\hat\theta)\;+\;\omega_b(\theta)\,\tilde b(\hat\theta;\theta),
\]
so that interim utility is $U_L(\hat\theta,\theta)=x(\hat\theta;\theta)+K(\theta)$ with $K$ absolutely continuous.
Assume the single-crossing condition in $x$:
\[
  \frac{\partial^2 U_L}{\partial\theta\,\partial x}(\hat\theta,\theta)\;\ge\;0\quad\text{for all }(\hat\theta,\theta).
\]
Under Assumption~\ref{ass:mlrp}, this holds because $\partial_\theta \tilde b(\hat\theta;\theta)\ge 0$ for nondecreasing $\beta$.
\end{assumption}

\begin{problem}[Leader’s program --- reduced form]\label{prob:leader}
\[
  \min_{T(\cdot),\,b(\cdot)}
    \; \E_\theta\!\Bigl[\,\E\bigl[\alpha\,\min\{\beta(\hat G),b(\theta)\}
    +\tfrac{\kappa}{2}\min\{\beta(\hat G),b(\theta)\}^{2}\bigm|\theta\bigr]
    + \gamma\,T(\theta)\Bigr]
\]
\[
\text{s.t.}\quad
\begin{cases}
  \text{(IC)} & V'(\theta)=\lambda_T'(\theta)\,T(\theta)+\omega_b\,\partial_\theta\tilde b(\theta;\theta)+K'(\theta),\\[0.5ex]
  \text{(IR)} & V(\theta)\ge\underline U,\quad\forall\theta, \\[0.5ex]
  \text{(LL)} & 0\le T(\theta),\; 0\le b(\theta)\le\bar b,\quad\forall\theta .
\end{cases}
\]
\end{problem}

\noindent\emph{Technical detail.} See Lemma~\ref{lem:mono-fixed} in Appendix~\ref{app:proofs}.

\begin{sidewaysfigure}
\centering
\resizebox{\textheight}{!}{%
\begin{tikzpicture}[
  >=Latex,
  node distance = 15mm and 38mm, 
  block/.style={rectangle, rounded corners, draw, align=left, fill=blue!3,
                text width=48mm, minimum height=8mm, inner sep=2mm},
  blockS/.style={block, font=\scriptsize},
  blockXS/.style={block, font=\scriptsize, text width=46mm}, 
  decision/.style={diamond, draw, aspect=2.3, align=center, fill=purple!3,
                   inner sep=1.2pt, text width=20mm},
  annot/.style={font=\scriptsize, align=left},
  every node/.style={font=\footnotesize}
]
\node (t0label) at (0,3.95) {\Large\bfseries $t=0$};
\node (t1label) at (6,3.95) {\Large\bfseries $t=1$};
\node (t2label) at (14.6,3.95) {\Large\bfseries $t=2$};

\node[block] (t0) at (0,2.45) {Province designs (commitment)\\
$T(\cdot),\ b(\cdot),\ \beta(\cdot)$};
\node[block,fill=gray!10] (t0alt) [below=of t0]
{Or: no commitment\\(discretion at $t=2$)};

\node[block] (choices) at (6,2.45)
{Local chooses $(e,q,I,D)$\\ s.t.\ $0\leq D\leq \bar{D}$};
\node[blockS] (gap) [below=of choices]
{$G=C(q,\theta)+(1-s)I+rD$\\
$-\,[R(e,\theta)+\tau+g+sI+D]+\varepsilon$};
\node[block] (signal) [below=of gap]
{$\hat G = G + \eta$};

\node[decision] (commit) at (14.6,2.85) {Commit-\\ment?};

\node[blockXS] (commitpath) [below left=16mm and 8mm of commit]
{Pay $p(\hat G,\hat\theta)$ with\\
$\displaystyle p=\mathbf{1}\{\hat G>0\}\min\{\beta(\hat G),\, b(\hat\theta),\, \hat G\}$};

\node[blockXS, fill=blue!6] (discretion) [below right=16mm and 15mm of commit]
{Choose $x$ to minimize\\
$\alpha x+\frac{\kappa}{2}x^2+\frac{\chi}{2}(\hat G-x)_+^2$\\[2pt]
$\Rightarrow\ \beta_{\mathrm{disc}}(\hat G)=
\Big[\frac{\chi \hat G-\alpha}{\kappa+\chi}\Big]_{[0,\min\{\bar{b},\hat G\}]}$\\
(TLC: threshold--linear--cap)};

\node[decision] (default) [below=26mm of commit] {$p<G$?};
\node[block,fill=red!5] (defaultyes) [below left=13mm and 12mm of default]
{Default / crisis\\ loss $\phi$};
\node[block,fill=green!5] (defaultno) [below right=13mm and 12mm of default]
{No default};

\draw[->,thick] (t0) -- (choices);
\draw[->,thick] (t0alt) to[out=15,in=195] (choices);

\draw[->,thick] (choices) -- (gap);
\draw[->,thick] (gap) -- (signal);

\draw[->,thick] (signal.east) to[out=5,in=220] (commit.west);

\draw[->,thick] (commit.south west) to[out=235,in=35] node[above,sloped]{yes} (commitpath.north east);
\draw[->,thick] (commit.south east) to[out=-55,in=145] node[above,sloped]{no}  (discretion.north west);

\draw[->,thick] (commitpath.south)  to[out=-80,in=180] (default.west);
\draw[->,thick] (discretion.south)  to[out=-100,in=0]  (default.east);

\draw[->,thick] (default) to[out=210,in=20]  (defaultyes);
\draw[->,thick] (default) to[out=-30,in=160] (defaultno);

\node[annot] at (6.0,1) {Matching: province pays $sI$, locality pays $(1-s)I$};
\node[annot,align=left] at (15.85,1.1) {Cap slack if $\beta(\hat G)<b(\hat\theta)$;\\ binding if $\geq$};

\draw[dashed,->] (gap.west) ++(-1.2,0) node[annot,left]{shock $\varepsilon$} -- (gap.west);
\draw[dashed,->] (signal.west) ++(-1.2,0) node[annot,left]{noise $\eta$} -- (signal.west);
\end{tikzpicture}
}
\caption{Timeline and decision flow (commitment vs.\ discretion with TLC).}
\end{sidewaysfigure}
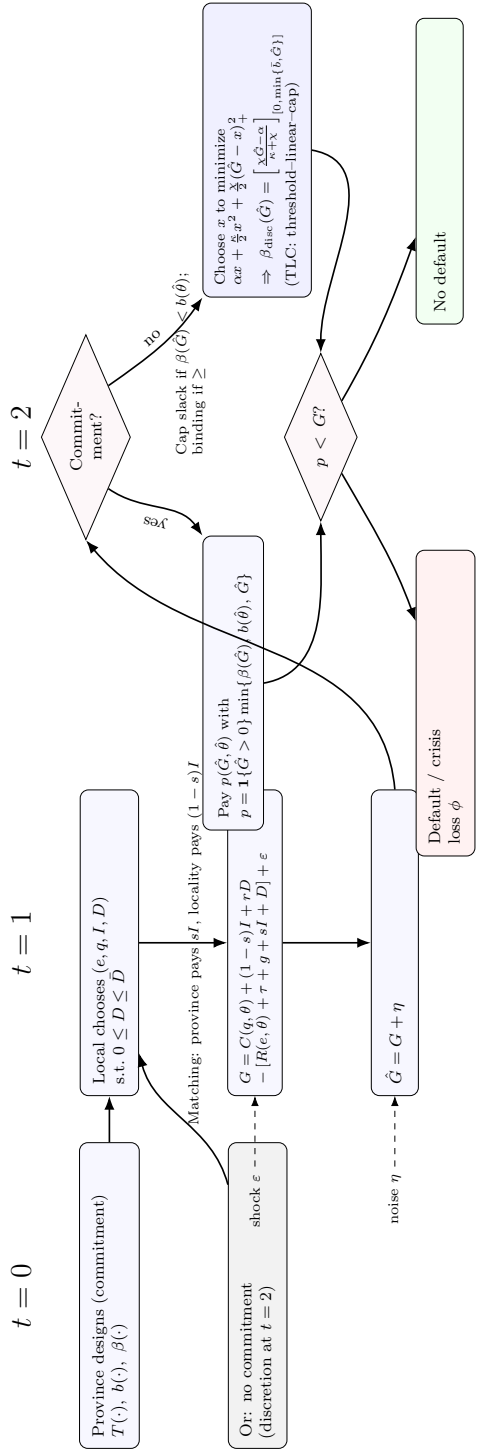

\subsection{Optimal Transfer Schedule}\label{sec:Tstar}

\noindent\emph{Technical detail.} See Lemma~\ref{lem:mono_cap} in Appendix~\ref{app:proofs}.

\begin{lemma}[Conditional cap--min calculus]\label{lem:capcalc}
Let $F_\beta(\cdot\mid\theta)$ be the c.d.f.\ of $\beta(\hat G)$ conditional on type with continuous density. For any cap $b\ge0$,
\[
  \frac{\partial}{\partial b}\E\bigl[\min\{\beta(\hat G),b\}\mid\theta\bigr]
  \;=\; \PP_\theta\!\bigl[\beta(\hat G)\ge b\bigr],\qquad
  \frac{\partial}{\partial b}\E\bigl[\min\{\beta(\hat G),b\}^{2}\mid\theta\bigr]
  \;=\; 2b\,\PP_\theta\!\bigl[\beta(\hat G)\ge b\bigr].
\]
\end{lemma}

\begin{proposition}[Closed-form optimal mechanism with statutory cap]\label{prop:opt_cap}
Let $f(\theta)>0$ and IFR on $[\underline\theta,\bar\theta]$, and let the provincial cost be
\(C(x)=\alpha x+\tfrac{\kappa}{2}x^{2}\) in the \emph{realized payout} $x$. In the threshold-$\beta$ baseline where $\lambda_T\equiv\omega_T$, the pointwise minimizer under $0\le b\le\bar b$ is
\begin{align}
  b^{\ast}(\theta)
    &=\min\!\Bigl\{\bar b,\;
        \max\!\Bigl\{0,\;
        \frac{\gamma\,\omega_b}{\kappa\,\omega_T}\,
        \frac{f(\theta)}{\barF(\theta)}
        -\frac{\alpha}{\kappa}
      \Bigr\}\Bigr\},
      \label{eq:b-star}\\[4pt]
  T^{\ast}(\theta)
    &=T^{\ast}(\theta^{\min})
      -\frac{\omega_b}{\omega_T}\,
      \Bigl[\tilde b^{\ast}(\theta)-\tilde b^{\ast}(\theta^{\min})\Bigr],
      \label{eq:T-star}
\end{align}
where $\tilde b^{\ast}(\theta)=\E\bigl[\min\{\beta(\hat G),b^{\ast}(\theta)\}\mid\theta\bigr]$, $\theta^{\min}=\inf\{\theta:\,b^{\ast}(\theta)>0\}$, and
\begin{equation}
  \theta^{\dagger} \;=\; \inf\bigl\{\theta:\, b^{\ast}(\theta)=\bar b \bigr\}.
  \label{eq:theta-dagger}
\end{equation}
\end{proposition}

\begin{proposition}[No-bailout knife-edge]\label{lem:nobailout}
Under Proposition~\ref{prop:opt_cap}'s conditions, a self-consistent
no bailout regime ($b^{\ast}(\theta)\equiv0$ for all $\theta$)
obtains iff
\[
  \boxed{\;\alpha\omega_T\;\ge\;\gamma\omega_b\,
          \sup_{\theta\in[\underline\theta,\bar\theta]}\haz(\theta)\;},
  \qquad \haz(\theta)=f(\theta)/\bar F(\theta).
\]
Otherwise the optimal cap is strictly positive on a set of positive measure.
\end{proposition}

\paragraph{IR normalization and LL implications.}
Normalize $V(\theta^{\min})=\underline U$ and note $\tilde b^{\ast}(\theta^{\min})=0$, whence \eqref{eq:T-star} gives $T^{\ast}(\theta^{\min})=0$. Because of the negative relation in \eqref{eq:T-star}, the limited-liability requirement $T\ge0$ implies that whenever $\tilde b^{\ast}(\theta)>0$ on some region, the optimal $T^{\ast}(\theta)$ is driven to the boundary $T=0$ there, shifting screening to $b(\cdot)$. Only when $\gamma$ is sufficiently small relative to $(\alpha,\kappa)$ can interior regions with $T>0$ arise.

\subsection{Comparative statics}\label{subsec:CS2}

Let $\haz(\theta)=f(\theta)/\bar F(\theta)$ and $\lambda_T=\omega_T$ in the baseline. The interior zero solves $h(\theta^{\min})=\alpha\,\lambda_T/(\gamma\,\omega_b)$. By the implicit function theorem,
\begin{align*}
  \frac{\partial\theta^{\min}}{\partial\alpha}
      &=\frac{1}{h'(\theta^{\min})}\,\frac{\lambda_T}{\gamma\,\omega_b}>0,\qquad
  \frac{\partial\theta^{\min}}{\partial\omega_b}
      =\frac{1}{h'(\theta^{\min})}\,\Bigl(-\frac{\alpha\,\lambda_T}{\gamma\,\omega_b^{2}}\Bigr)<0,\\[6pt]
  \frac{\partial\theta^{\min}}{\partial\lambda_T}
      &=\frac{1}{h'(\theta^{\min})}\,\frac{\alpha}{\gamma\,\omega_b}>0,\qquad
  \frac{\partial\theta^{\min}}{\partial\gamma}
      =\frac{1}{h'(\theta^{\min})}\,\Bigl(-\frac{\alpha\,\lambda_T}{\gamma^{2}\,\omega_b}\Bigr)<0.
\end{align*}

For the interior cap $b_{\max}=(\gamma\omega_b/\lambda_T-\alpha)/\kappa$,
\[
  \frac{\partial b_{\max}}{\partial\kappa}
      =-\,\frac{1}{\kappa^{2}}
         \Bigl(\frac{\gamma\omega_b}{\lambda_T}-\alpha\Bigr)<0,\quad
  \frac{\partial b_{\max}}{\partial\lambda_T}
      =-\,\frac{\gamma\omega_b}{\kappa\lambda_T^{2}}<0,\quad
  \frac{\partial b_{\max}}{\partial\gamma}
      =\frac{\omega_b}{\kappa\lambda_T}>0.
\]

\subsection{Welfare Property}\label{subsec:CS2b}

\begin{proposition}
[Second--best efficiency]\label{prop:welfare}
Under Assumptions \ref{ass:primitives}--\ref{ass:sc},  
the mechanism in Proposition~\ref{prop:opt_cap}
maximizes the expected sum of provincial and municipal utilities
among all IC--IR--LL allocations.
\end{proposition}

\section{Policy Implications}\label{sec:policy}

Even though the present paper is purely theoretical, its closed--form solution
delivers lessons that speak directly to provincial practice in
NL and other fiscally stretched jurisdictions.

\subsection{Design principles}

\begin{enumerate}[label=\textbf{P\arabic*},leftmargin=*]

\item \textbf{Codify a \emph{triple-zone} rule.}  
      Proposition~\ref{prop:opt_cap} together with \eqref{eq:theta-dagger}
implies a simple menu in the baseline: \textit{(i)} no transfer when the reported type is below
      $\theta^{\min}$;  
      \textit{(ii)} a flat-to-rising cap on $[\theta^{\min},\,\theta^{\dagger}]$ following \eqref{eq:b-star};  
      \textit{(iii)} a flat cap once type exceeds $\theta^{\dagger}$.  Under a discretionary $t{=}2$ rule (Appendix~\ref{subsec:disc}), the realized payout becomes \emph{threshold--linear--cap} in $\hat G$.

\item \textbf{A single inequality decides whether bailouts survive.}  
      Under the no bailout candidate, Proposition~\ref{lem:nobailout} shows that bailouts disappear when
      $\displaystyle \alpha\omega_T \ge \gamma\omega_b \cdot \sup_{\theta} \haz(\theta)$.

\item \textbf{Front-load under softness (discretion).}  
      When the $t{=}2$ payout rule has an interior linear segment (Appendix~\ref{subsec:disc}), the effective grant weight falls below $\omega_T$ proportionally to the slope and the probability of being on that linear branch (Lemma~\ref{lem:lambdaT}). Softer rescue (steeper or more likely linear branch) thus calls for more front-loaded $T$; cheaper provincial financing ($\gamma\downarrow$) pushes in the opposite direction.

\item \textbf{Make the cap bite by increasing \(\kappa\).}  
      The optimal cap $b_{\max}=(\gamma\omega_b/\lambda_T-\alpha)/\kappa$ is inversely proportional to $\kappa$.
\end{enumerate}

\subsection{Political-economy robustness}

\begin{enumerate}[label=\textbf{E\arabic*},leftmargin=*]

\item \textbf{Credibility of caps.}  
      Without a hard legal ceiling, expectations rise, the discretion-based linear branch becomes more likely, the effective $\lambda_T$ falls, and effort weakens.

\item \textbf{Information precision.}
      Better accounting (lower $\sigma_\eta$) reduces the measure of states on which the cap is slack at high signals and reinforces the effectiveness of the cap.

\item \textbf{Vertical externalities}.  
      The knife-edge in Proposition~\ref{lem:nobailout} offers a quantitative stress-test benchmark.
\end{enumerate}

\subsection{Limited-liability regions}
\noindent From \eqref{eq:T-star}, $T^{\ast}(\theta)$ is (weakly) decreasing in $\tilde b^{\ast}(\theta)$.
Hence on any type region where $\tilde b^{\ast}(\theta)>0$, the grant LL constraint ($T\ge0$) typically binds, pushing $T^{\ast}(\theta)$ to $0$ and shifting screening to $b(\cdot)$.
Consequently, interior $T^{\ast}>0$ arises only (i) on the no bailout region where $b^{\ast}=0$ (i.e.\ $\theta<\theta^{\min}$), or (ii) on ironed segments of the virtual weight when ironing is required under IFR.

\medskip
\noindent When $\gamma$ falls, $b^{\ast}$ rises pointwise on its interior segment and the region with $T^{\ast}=0$ expands, unless $\theta^{\min}$ shifts left enough to increase the no bailout mass.

\subsection{Limitations and extensions}

\begin{enumerate}[label=\textbf{L\arabic*},leftmargin=*]
  \item \textbf{Static benchmark.}  
        A multi-period version with learning would allow reputation-based caps
        and dynamic debt paths.
  \item \textbf{Heterogeneous \(\omega_b\).}  
        Allowing \(\omega_b(\theta)\) to vary with political or administrative
        quality could sharpen incentives (Appendix~\ref{app:omega-variable}).
  \item \textbf{Federal--provincial layer}.  
        Embedding one more tier would capture Ottawa’s equalization backstop
        and test whether the knife-edge extends upward in the federation.
\end{enumerate}

\subsection{Implementation checklist and robustness}
\label{subsec:checklist}
\paragraph{A. Minimal implementation checklist.}
\begin{enumerate}[leftmargin=*]
\item \textbf{Fix primitives:} pick a parametric family for $F(\theta)$ with IFR (e.g.\ log-logistic/exponential tail) and specify $(\alpha,\kappa,\gamma,\omega_T,\omega_b,\bar b)$ from budget documents or ranges.
\item \textbf{Compute hazard:} $\haz(\theta)=f(\theta)/\barF(\theta)$ and locate $\theta^{\min}$ from $h(\theta^{\min})=\alpha\,\omega_T/(\gamma\,\omega_b)$; set $\theta^{\dagger}$ by $b^\ast(\theta)=\bar b$ as in~\eqref{eq:theta-dagger}.
\item \textbf{Construct the menu:} $b^\ast(\theta)$ by \eqref{eq:b-star} and $T^\ast(\theta)$ by \eqref{eq:T-star}, then codify the triple-zone rule (no transfer / rising cap / flat cap).
\item \textbf{Audit rule:} if $t{=}2$ discretion applies, use Appendix~\ref{subsec:disc} to parameterize the linear slope and adjust the effective grant weight $\lambda_T$ accordingly (cf.\ Lemma~\ref{lem:lambdaT}).
\end{enumerate}

\paragraph{B. Robustness flags.}
The design principles rely on three modeling choices:
(i) IFR for $F$ (needed only for monotonicity/ironing);
(ii) a nondecreasing, threshold-like $\beta$ (administratively common);
(iii) continuous signal noise (so boundary terms vanish and $\lambda_T\equiv\omega_T$ under threshold $\beta$).
Under a discretionary linear branch, $\lambda_T$ is scaled down by the branch slope times its probability; see Lemma~\ref{lem:lambdaT} and Appendix~\ref{subsec:disc}.

\section{Conclusion}\label{sec:conclusion}
This paper recasts provincial--municipal rescues as a two-instrument screening problem.
Four takeaways emerge:
\begin{enumerate}[leftmargin=*]
\item \textbf{One-dimensional reduction.} Under convexity and MLRP, the three-stage environment folds into a single-index screen in $(T,b)$ (Proposition~\ref{prop:reduction}).
\item \textbf{Closed-form cap rule.} With linear--quadratic costs and a statutory cap, the IC--IR--LL optimum is \emph{threshold--cap}, with cutoffs $(\theta^{\min},\theta^{\dagger})$ pinned down by the hazard $\haz(\theta)$ (Proposition~\ref{prop:opt_cap}).
\item \textbf{Unified regime test.} A self-consistent no bailout regime obtains iff $\alpha\omega_T \ge \gamma\omega_b\cdot \sup_\theta \haz(\theta)$, providing a single marginal-cost benchmark (Proposition~\ref{lem:nobailout}).
\item \textbf{Discretion vs.\ commitment.} Without commitment at $t{=}2$, the realized rule becomes \emph{threshold--linear--cap}; the interior slope lowers the effective grant weight and strengthens soft budget incentives (Appendix~\ref{subsec:disc} and Lemma~\ref{lem:lambdaT}).

\end{enumerate}
\noindent\textit{Limitations and next steps.} (i) Our baseline adopts IFR and threshold-like $\beta$; when $\haz(\theta)$ is nonmonotone, standard ironing applies and preserves implementability. (ii) A stylized calibration---even with public ranges for $(\alpha,\kappa,\gamma,\omega_T,\omega_b,\bar b)$---would illustrate the triple-zone geometry and aid policy communication. (iii) Extending to a two-dimensional type or audit manipulation cost is feasible and would test the scope of the hazard-based cutoffs.

\newpage
\setcounter{table}{0}
\begin{appendices}

\section{Symbols used in the model and empirical discussion}\label{tab:symbols}
\begin{table}[H]
\centering
\small
\caption{Symbols used in the model and empirical discussion.}
\label{tab:notation}
\begin{tabular}{@{}lp{9.2cm}@{}}
\toprule
\textbf{Symbol} & \textbf{Description} \\
\midrule
$i$ & Local jurisdiction index. \\
$\theta$ & Fiscal need / gap type (higher $=$ weaker tax base, larger need). \\
$f(\theta),F(\theta)$ & Density and c.d.f.\ of types on $[\underline\theta,\bar\theta]$. \\
$\barF(\theta)$ & Survivor function $1-F(\theta)$. \\
$e$ & Local revenue effort. \\
$R(e,\theta)$ & Own-source revenue function. \\
$q$ & Level of basic services. \\
$C(q,\theta)$ & Cost to produce $q$ given $\theta$. \\
$\tau$ & Unconditional operating grant (MOG); maps to $T$ in the model. \\
$g$ & Predictable capital transfer (CCBF). \\
$s$ & Provincial cost-share rate for capital $I$. \\
$I$ & Local capital investment. \\
$D$ & New debt (subject to approval); $\bar D$ debt limit. \\
$r$ & Debt-service factor on $D$. \\
$G$ & Ex post fiscal gap before payout. \\
$\beta(\hat G)$ & Signal-based payout rule at $t{=}2$. \\
$b(\theta)$ & Type-based cap at $t{=}0$. \\
$p(\hat G,\hat\theta)$ & Realized payout $\mathbf 1\{\hat G>0\}\min\{\beta(\hat G),b(\theta),\hat G\}$. \\
$T(\theta)$ & Ex-ante grant schedule. \\
$\tilde b(\hat\theta;\theta)$ & Expected payout under cap: $\E[\min\{\beta(\hat G),b(\hat\theta)\}\mid\theta]$. \\
$\omega_T,\,\omega_b$ & Marginal utilities of $T$ and realized payout. \\
$\alpha,\kappa$ & Payout cost parameters: $C(x)=\alpha x+\frac{\kappa}{2}x^{2}$. \\
$\lambda_T(\theta)$ & Effective weight on $T$ in screening: $\omega_T$ (threshold baseline). \\
$U_L$ & Local government utility (interim). \\
$U_P$ & Province’s (negative) expected cost. \\
$V(\theta)$ & Truthful utility $U_L(\theta,\theta)$. \\
$\underline U$ & Reservation utility (IR constraint). \\
$\theta^{\min},\theta^{\dagger}$ & Lower/upper cutoffs for the optimal cap. \\
$b_{\max}$ & Interior bailout cap level. \\
$\rho$ & Discount factor in the repeated game. \\
$\rho^{\star}$ & Critical discount factor. \\
$\pi$ & Ex-ante default-coverage probability (descriptive). \\
$\delta$ & Ex-ante default probability. \\
$\phi$ & Welfare loss to residents under unresolved gap. \\
$\bar b$ & Statutory cap on per-period bailout. \\
$\eta$ & Audit/report noise. \\
$B(q),\;\Gamma(I)$ & Utility/benefit from service level $q$ and investment $I$. \\
$\chi$ & Convex loss parameter in discretionary rescue (Appendix~\ref{subsec:disc}). \\
\bottomrule
\end{tabular}
\end{table}
\newpage

\section{Discretionary rescue at $t=2$ and backward induction}\label{subsec:disc}

\paragraph{Provincial problem at $t=2$ (no commitment).}
Suppose the Province cannot commit to $\beta$ at $t=0$ and instead chooses 
a payout $x$ at $t=2$ after observing the noisy gap signal $\hat G$. 
For tractability, let the loss from an unresolved residual gap $(\hat G-x)_{+}$ be convex:
\[
  L\bigl((\hat G-x)_{+}\bigr) \;=\; \frac{\chi}{2}\,(\hat G-x)_{+}^{2}, 
  \qquad \chi>0.
\]
The Province solves
\[
  \min_{\,0\le x \le \min\{\bar b,\,\hat G\}}\;
      \alpha x \;+\; \frac{\kappa}{2}x^{2}
      \;+\; L\bigl((\hat G-x)_{+}\bigr).
\]
When $0<x<\min\{\bar b,\,\hat G\}$ the FOC is
$\alpha+\kappa x - \chi(\hat G-x) = 0$, hence
\[
  \beta^{\text{disc}}(\hat G)
   \;=\; \left[\, \frac{\chi\,\hat G - \alpha}{\kappa+\chi} \,\right]_{[\,0,\,\bar b\,]}.
\]
Therefore $\beta^{\text{disc}}$ is \emph{threshold--linear--cap} in $\hat G$.

\paragraph{Backward induction to $t=1$.}
Municipalities at $t=1$ anticipate $\beta^{\text{disc}}(\cdot)$ and choose effort accordingly.
On the interior linear branch where $\beta^{\text{disc}}{}'(\hat G)=\chi/(\kappa+\chi)$, 
the default-probability component in \eqref{eq:FOC_e_correct} is scaled by $\kappa/(\kappa+\chi)$
and there is an additional marginal-rescue term $-\omega_b\,\E[\beta^{\text{disc}}{}'(\hat G)\mathbf1\{\beta^{\text{disc}}<b\}]$.

\paragraph{Discussion.}
This discretionary benchmark microfounds a threshold--linear--cap rule at $t=2$ and shows how the slope filters into 
\eqref{eq:FOC_e_correct}, strengthening the soft budget moral-hazard channel. 
Our commitment baseline avoids time inconsistency by fixing $\beta$ at $t=0$; 
the discretion variant is useful as a robustness check.

\section{Variable marginal utility of bailouts}\label{app:omega-variable}

When the marginal utility of a realized bailout, $\omega_b(\theta)$, varies
across jurisdictions—for instance because political pressure is stronger for
small communities—the first-order condition for the optimal cap becomes
\[
  \omega_b(\theta)\;=\;\alpha\;+\;\kappa\,b^{\ast}(\theta),
\]
so that the linear segment in~\eqref{eq:b-star} reads
\[
  b^{\ast}(\theta)
  \;=\;\frac{\omega_b(\theta)-\alpha}{\kappa},
  \qquad
  0\;\le\;b^{\ast}(\theta)\;\le\;\bar b.
\]
\textit{Implication.} As long as $\omega_b(\theta)$ is weakly increasing in
$\theta$, the cap schedule remains monotone and retains the
\emph{threshold–linear–cap} geometry under the discretionary benchmark.  The slope may now vary with type; for
empirical calibration one needs an estimate of $\omega_b(\theta)$, e.g.\ from
survey weights or past voting patterns.

\section{Single-crossing and monotonicity details}\label{app:sc-proof}
With $U_L(\hat\theta,\theta)=x(\hat\theta;\theta)+K(\theta)$ and $\partial^2 U_L/\partial\theta\,\partial x\ge0$,
the Spence--Mirrlees single-crossing property implies standard IC inequalities:
for any $\theta>\hat\theta$,
\[
  \bigl[U_L(\theta,\theta)-U_L(\hat\theta,\theta)\bigr]
  \;\ge\;
  \bigl[U_L(\theta,\hat\theta)-U_L(\hat\theta,\hat\theta)\bigr].
\]
Since $K$ cancels, this reduces to
$x(\theta;\theta)-x(\hat\theta;\theta)\ge x(\theta;\hat\theta)-x(\hat\theta;\hat\theta)$.
By letting reports be truthful on the RHS, we get $x(\theta)\ge x(\hat\theta)$, hence monotonicity.
When the virtual term $\bigl(\gamma\omega_b/\lambda_T(\theta)\bigr)\haz(\theta)$ fails to be increasing,
standard ironing (\`a la Myerson) delivers a nondecreasing ironed index.

\section{Regularity for differentiation under the expectation}\label{app:regularity-diff}
We justify the steps leading to \eqref{eq:FOC_e_correct} and Lemma~\ref{lem:margprob}.

\paragraph{Dominated convergence / Leibniz rule.}
Assume: (i) $\eta$ has a continuous density $f_\eta$ with bounded tails; (ii) $\beta$ is piecewise $C^1$ with slope in $[0,1)$ and bounded image; (iii) $R'_e(e,\theta)$ is continuous and locally bounded uniformly in $e$ on compact sets. Then, for any integrable function $g(\hat G,e)$ that is piecewise $C^1$ in $e$ and dominated by an integrable envelope, we may differentiate inside the expectation by dominated convergence / Leibniz’s rule:
\[
\frac{\partial}{\partial e}\,\E\bigl[g(\hat G,e)\bigr]
=\E\Bigl[\frac{\partial}{\partial e}g(\hat G,e)\Bigr].
\]

\paragraph{Indicators and boundary sets.}
For events of the form $\{\hat G-\beta(\hat G)>0\}$, the boundary $\{\hat G-\beta(\hat G)=0\}$ has Lebesgue measure zero because $\eta$ has a density and $\beta$ is a.e.\ differentiable with bounded slope; hence the derivative of the indicator contributes no boundary term.
On threshold rules, $\beta'(\hat G)=0$ a.e., so the marginal-rescue term vanishes, yielding the expressions stated in Lemma~\ref{lem:margprob} and \eqref{eq:FOC_e_correct}.

\section{Technical Lemmas and Proofs}\label{app:proofs}

\noindent\textit{Note (implementable payout).} Throughout, the realized payout is
$p(\hat G,\hat\theta)=\mathbf 1\{\hat G>0\}\min\{\beta(\hat G),b(\hat\theta),\hat G\}$.
On the cap–slack and positive–signal set where $\beta(\hat G)<b(\hat\theta)$,
all derivatives below coincide with those under $p=\beta(\hat G)$.
When the $\min$ picks $\hat G$, boundary sets have Lebesgue measure zero under continuous noise, so the derivative contributions vanish a.e.

\begin{lemma}[Effort independence from report]\label{lem:e-indep}
Under Assumptions~\ref{ass:primitives}, \ref{ass:regularity} and
\ref{ass:beta-threshold}, the interior first-order condition
\eqref{eq:FOC_e_correct} can be rewritten
\[
R'_e\!\bigl(e^{\ast}(\theta),\theta\bigr)\bigl\{1+\varphi\,\Lambda\bigr\}=\phi,
\qquad\Lambda=\E\bigl[f_\eta(\hat G-\beta(\hat G))\bigr].
\]
All terms on the right depend only on the \emph{true} type~$\theta$; hence
$e^\star(\theta)$ is independent of the reported~$\hat\theta$ up to boundary-density terms on the cap-binding tail.
\end{lemma}

\begin{lemma}[Marginal default probability on the signal branch]\label{lem:margprob}
Under Assumptions~\ref{ass:primitives}--\ref{ass:regularity}, let $\delta=\PP_\theta\!\bigl[p(\hat G,\hat\theta)<G\bigr]$. On the set where $\beta(\hat G)<b(\hat\theta)$ (cap slack),
\[
  \frac{\partial \delta}{\partial e}
  \;=\;
  -\,R'_e(e,\theta)\,\E\!\left[f_\eta\!\bigl(\hat{G}-\beta(\hat{G})\bigr)\,\bigl(1-\beta'(\hat G)\bigr)\right],
\]
and for \emph{threshold} rules (piecewise constant $\beta$), $\beta'(\hat G)=0$ a.e., so
$\displaystyle \frac{\partial \delta}{\partial e}=-R'_e(e,\theta)\,\E\!\left[f_\eta\!\bigl(\hat{G}-\beta(\hat{G})\bigr)\right]$.
\end{lemma}

\begin{proof}[Proof of Lemma~\ref{lem:lambdaT}]
Recall $\hat G=G+\eta$ and $\partial_T \hat G=\partial_T G=-1$. Write
\[
p(\hat G,\hat\theta)=\mathbf 1\{\hat G>0\}\min\{\beta(\hat G),b(\hat\theta),\hat G\}.
\]
Since $\beta(0)=0$ and $\beta'\in[0,1)$, for $\hat G\ge 0$ we have $\beta(\hat G)\le \hat G$. Hence on $\{\hat G>0\}$ and cap–slack $\{\beta(\hat G)<b(\hat\theta)\}$, the minimum is $\beta(\hat G)$ and
\[
\partial_T p = \beta'(\hat G)\,\partial_T \hat G = -\,\beta'(\hat G).
\]
On the cap-binding set $\{\beta(\hat G)\ge b(\hat\theta)\}$, $p=b(\hat\theta)$ so $\partial_T p=0$. Therefore,
\[
\frac{\partial}{\partial T}\E[p(\hat G,\hat\theta)\mid\theta]
= -\,\E\!\left[\beta'(\hat G)\,\mathbf 1\{\beta(\hat G)<b(\hat\theta)\}\,\middle|\,\theta\right]+\text{boundary terms}.
\]
The boundary terms arise only when \(arg min\{\beta(\hat G),\, b(\hat\theta),\, \hat G\}\) changes; 
since \(\eta\) has a continuous density and \(\beta\) is a.e.\ differentiable with slope \(<1\), 
those switch sets have Lebesgue measure zero and their contribution vanishes under dominated convergence. Hence $\lambda_T(\theta)=\omega_T-\omega_b\,\partial_T\E[p|\theta]$ reduces to the stated expressions; in particular, for threshold $\beta$ we obtain $\lambda_T(\theta)\equiv \omega_T$. \qedhere
\end{proof}

\begin{lemma}[Monotonicity of the allocation index]\label{lem:mono-fixed}
Under IC and Assumption~\ref{ass:sc}, the implemented allocation index
\[
  x(\theta)\;=\;\lambda_T(\theta)\,T(\theta)+\omega_b(\theta)\,\tilde b(\theta;\theta)
\]
is weakly increasing in $\theta$. When ironing is required under IFR, the ironed allocation is nondecreasing.
\end{lemma}

\begin{lemma}[Monotonicity under IFR and caps]\label{lem:mono_cap}
If $\lambda_T$ and $\omega_b$ are locally constant and $f/\barF$ is
increasing (IFR), then the optimal cap $b^{\ast}(\theta)=\min\{\bar b,\max\{0,\tilde b(\theta)\}\}$ is weakly increasing,
where $\tilde b(\theta)=\frac{\gamma\omega_b}{\kappa\lambda_T}\frac{f(\theta)}{\barF(\theta)}-\frac{\alpha}{\kappa}$.
If $\lambda_T(\theta)$ varies, a sufficient condition is that
$\lambda_T(\theta)$ is weakly decreasing and $f(\theta)/\barF(\theta)$ is
increasing; otherwise, apply standard ironing on the virtual term $\frac{\gamma\omega_b}{\lambda_T(\theta)}\frac{f(\theta)}{\barF(\theta)}$.
\end{lemma}

\begin{proof}[Proof of Eq.~\eqref{eq:FOC_e_correct} (first-order condition for $e$)]
Fix $(\tau,s,\bar D,\beta,b)$ and true type $\theta$. The municipality’s $t{=}1$ objective as a function of $e$ (dropping terms independent of $e$) is
\[
\Phi(e)\;=\;\E\!\Big[R(e,\theta)-\phi e\;-\;\varphi\,\mathbf 1\{p(\hat G,\hat\theta)<G\}\;+\;\omega_b\,p(\hat G,\hat\theta)\Big],
\]
with $\,\hat G=G(e)+\eta\,$ and $\,G(e)=C(\cdot)-[R(e,\theta)+\tau+g]+\ldots\,$ so that $\partial_e \hat G=\partial_e G=-R'_e(e,\theta)$.

\smallskip
\noindent\textbf{Step 1 (Justifying differentiation under $\E$).}
By Assumptions~\ref{ass:regularity}--\ref{ass:beta-threshold}, $\eta$ has a continuous density $f_\eta$ with bounded tails, $\beta$ is piecewise $C^1$ with slope in $[0,1)$ and bounded image, and $R'_e$ is continuous and locally bounded. Hence all integrands below admit a uniform integrable envelope, so dominated convergence / Leibniz rule applies and we may interchange $\partial_e$ and $\E$.

\smallskip
\noindent\textbf{Step 2 (Derivative of the default indicator).}
Define $H(e,\eta)=\hat G-\beta(\hat G)$. On the set where $\beta$ is differentiable,
\[
\partial_e H(e,\eta)\;=\;(\partial_e \hat G)\,\bigl(1-\beta'(\hat G)\bigr)\;=\;-R'_e(e,\theta)\,\bigl(1-\beta'(\hat G)\bigr).
\]
Approximate the Heaviside $\mathbf 1\{u>0\}$ by smooth $s_n(u)$ with $s_n'\to \delta_0$ in the sense of distributions, and apply dominated convergence:
\[
\frac{\partial}{\partial e}\E\!\big[\mathbf 1\{H(e,\eta)>0\}\big]
=\lim_{n\to\infty}\E\!\big[s_n'(H)\,\partial_e H\big]
= \E\!\left[\delta_0\!\big(H\big)\,\partial_e H\right].
\]
Since $H=\hat G-\beta(\hat G)$ is a monotone $C^1$ transformation of $\hat G$ with slope $1-\beta'(\hat G)\in(0,1]$ a.e., the density of $H$ at $0$ equals $f_\eta\!\big(\hat G-\beta(\hat G)\big)$ a.e. Hence
\[
\frac{\partial}{\partial e}\E\!\big[\mathbf 1\{p(\hat G,\hat\theta)<G\}\big]
=\frac{\partial}{\partial e}\E\!\big[\mathbf 1\{H>0\}\big]
= -\,R'_e(e,\theta)\,\E\!\Big[f_\eta\!\big(\hat G-\beta(\hat G)\big)\,\bigl(1-\beta'(\hat G)\bigr)\Big],
\]
where we have used that on the cap–slack set $\{\,\beta(\hat G)<b(\hat\theta)\,\}$ the event $\{p<G\}$ coincides with $\{H>0\}$, while on the cap-binding set the boundary-density contribution is $O(\varepsilon)$ by Assumption~\ref{ass:cap-slack} and does not affect the sign/comparative statics.

\smallskip
\noindent\textbf{Step 3 (Derivative of the realized payout).}
Write $p(\hat G,\hat\theta)=\mathbf 1\{\hat G>0\}\min\{\beta(\hat G),b(\hat\theta),\hat G\}$. Since $\beta(0)=0$ and $\beta'\in[0,1)$, we have $\beta(\hat G)\le \hat G$ for all $\hat G\ge 0$; thus on $\{\hat G>0\}$ and cap–slack $\{\beta(\hat G)<b(\hat\theta)\}$, $p=\beta(\hat G)$ and
\[
\partial_e p \;=\; \beta'(\hat G)\,\partial_e \hat G \;=\; -\,\beta'(\hat G)\,R'_e(e,\theta).
\]
On the cap-binding set $p=b(\hat\theta)$ so $\partial_e p=0$; hence
\[
\frac{\partial}{\partial e}\E[p(\hat G,\hat\theta)]
= -\,R'_e(e,\theta)\,\E\!\Big[\beta'(\hat G)\,\mathbf 1\{\beta(\hat G)<b(\hat\theta)\}\Big]\;+\;O(\varepsilon).
\]

\smallskip
\noindent\textbf{Step 4 (FOC).}
Collecting terms,
\[
\Phi'(e)=R'_e(e,\theta)-\phi-\varphi\,\frac{\partial}{\partial e}\E[\mathbf 1\{p<G\}]+\omega_b\,\frac{\partial}{\partial e}\E[p].
\]
Using the expressions above and canceling the common factor $R'_e(e,\theta)$, the $O(\varepsilon)$ terms vanish by Assumption~\ref{ass:cap-slack}, and the first-order condition $\Phi'(e^\ast)=0$ becomes
\[
R'_e\!\bigl(e^{\ast}(\theta),\theta\bigr)\Big\{1+\varphi\,\E\!\big[f_\eta(\hat G-\beta(\hat G))\bigl(1-\beta'(\hat G)\bigr)\big]-\omega_b\,\E\!\big[\beta'(\hat G)\,\mathbf 1\{\beta(\hat G)<b(\hat\theta)\}\big]\Big\}=\phi,
\]
which is Eq.~\eqref{eq:FOC_e_correct}. \end{proof}

\begin{proof}[Proof of Lemma~\ref{lem:margprob}]
Let $\delta(e)=\PP_\theta[p(\hat G,\hat\theta)<G]$. On the cap–slack event $\{\beta(\hat G)<b(\hat\theta)\}$ we have $\{p<G\}=\{\hat G-\beta(\hat G)>0\}= \{H>0\}$. Repeating the mollifier argument in the proof of Eq.~\eqref{eq:FOC_e_correct},
\[
\delta'(e)=\frac{\partial}{\partial e}\E[\mathbf 1\{H>0\}]=\E\!\left[\delta_0(H)\,\partial_e H\right]
= -\,R'_e(e,\theta)\,\E\!\left[f_\eta\!\big(\hat G-\beta(\hat G)\big)\,\bigl(1-\beta'(\hat G)\bigr)\right].
\]
For \emph{threshold} $\beta$ we have $\beta'(\hat G)=0$ a.e., hence
\[
\delta'(e)=-\,R'_e(e,\theta)\,\E\!\left[f_\eta\!\big(\hat G-\beta(\hat G)\big)\right].
\]
On the cap-binding set, the event $\{p<G\}$ becomes $\{b(\hat\theta)<G\}$ and contributes only a boundary-density term with probability $O(\varepsilon)$, which does not alter the formula. \qedhere
\end{proof}

\begin{proof}[Proof of Lemma~\ref{lem:capcalc}]
Let $Y=\beta(\hat G)$ with c.d.f.\ $F_\beta(\cdot\mid\theta)$ and continuous density. Then
\[
\E[\min\{Y,b\}\mid\theta]=\int_{0}^{\infty}\min\{y,b\}\,\dd F_\beta(y\mid\theta)
=\int_{0}^{b}y\,\dd F_\beta+\int_{b}^{\infty}b\,\dd F_\beta.
\]
Differentiating w.r.t.\ $b$ and using the fundamental theorem for Lebesgue integrals gives
\[
\partial_b \E[\min\{Y,b\}\mid\theta]=\PP_\theta[Y\ge b].
\]
Similarly,
\[
\E[\min\{Y,b\}^2\mid\theta]=\int_{0}^{b}y^2\,\dd F_\beta+\int_{b}^{\infty}b^2\,\dd F_\beta
=\int_{0}^{b}\PP_\theta[Y\ge t]\,(2t)\,\dd t,
\]
so $\partial_b \E[\min\{Y,b\}^2\mid\theta]=2b\,\PP_\theta[Y\ge b]$. \qedhere
\end{proof}

\begin{proof}[Proof of Proposition~\ref{prop:reduction} (reduction to quasi-linearity)]
Fix $(\tau,s,\bar D,\beta)$ and true type $\theta$. Minimizing over $(q,I,D)$ yields reduced cost $C_0(\theta)$ independent of the report. By Assumption~\ref{ass:report-invariance}, the interior $e^\ast(\theta)$ is report-independent up to $O(\varepsilon)$ boundary terms. Hence interim utility can be written as
\[
U_L(\hat\theta,\theta)=\lambda_T(\theta)\,T(\hat\theta)+\omega_b\,\tilde b(\hat\theta;\theta)+K(\theta),
\]
with $\lambda_T(\theta)=\omega_T-\omega_b\,\partial_T\E[p(\hat G,\hat\theta)\mid\theta]$ and $\tilde b(\hat\theta;\theta)=\E[\min\{\beta(\hat G),b(\hat\theta)\}\mid\theta]$. By Lemma~\ref{lem:lambdaT}, under threshold $\beta$ and continuous noise the boundary terms vanish and $\partial_T\E[p|\theta]=0$, hence $\lambda_T(\theta)=\omega_T$ a.e. This delivers the quasi-linear reduced form \eqref{eq:UL-reduced}. \qedhere
\end{proof}

\begin{proof}[Proof of Lemma~\ref{lem:mono-fixed}]
Let $x(\hat\theta;\theta)=\lambda_T(\theta)T(\hat\theta)+\omega_b(\theta)\tilde b(\hat\theta;\theta)$ and $U_L(\hat\theta,\theta)=x(\hat\theta;\theta)+K(\theta)$. By Assumption~\ref{ass:sc}, $\partial^2 U_L/\partial\theta\,\partial x\ge 0$ (Spence–Mirrlees). Suppose, towards a contradiction, that there exist $\theta_2>\theta_1$ with $x(\theta_2)<x(\theta_1)$. Then IC implies
\[
U_L(\theta_2,\theta_2)\ge U_L(\theta_1,\theta_2),\qquad
U_L(\theta_1,\theta_1)\ge U_L(\theta_2,\theta_1).
\]
Subtracting and using single crossing yields $x(\theta_2)\ge x(\theta_1)$, a contradiction. Hence $x(\theta)$ is weakly increasing. When ironing is needed (IFR with nonmonotone virtual term), the ironed allocation preserves nondecreasingness. \qedhere
\end{proof}

\begin{proof}[Proof of Lemma~\ref{lem:mono_cap}]
Fix $\theta$ and consider an incremental increase of $b(\theta)$ by $\dd b$ while holding $b$ elsewhere fixed. By Lemma~\ref{lem:capcalc}, the marginal increase in the Province’s expected cost at type $\theta$ equals
\[
(\alpha+\kappa b(\theta))\,\PP_\theta[\beta(\hat G)\ge b(\theta)]\,\dd b.
\]
By Myerson’s envelope for direct mechanisms, the marginal (virtual) benefit from relaxing the cap at $\theta$ equals
\[
\frac{\gamma}{\lambda_T(\theta)}\,\omega_b(\theta)\,\haz(\theta)\,\PP_\theta[\beta(\hat G)\ge b(\theta)]\,\dd b,
\]
where $\haz(\theta)=f(\theta)/\bar F(\theta)$ under IFR. Equating marginal cost and benefit cancels the common tail probability and yields the pointwise KKT condition
\[
\alpha+\kappa b(\theta)=\frac{\gamma\,\omega_b(\theta)}{\lambda_T(\theta)}\,\haz(\theta).
\]
If $\lambda_T,\omega_b$ are locally constant, this implies $b(\theta)=\frac{1}{\kappa}\Big(\frac{\gamma\omega_b}{\lambda_T}\haz(\theta)-\alpha\Big)$, which is increasing in $\theta$ because $h$ is increasing under IFR. Projection onto $[0,\bar b]$ preserves weak monotonicity. If $\lambda_T(\theta)$ varies with $\theta$, a sufficient condition for the RHS to be weakly increasing is that $\lambda_T$ be weakly decreasing while $h$ is weakly increasing; otherwise standard ironing of the virtual term $\frac{\gamma\omega_b}{\lambda_T(\theta)}\haz(\theta)$ restores a nondecreasing $b^\ast(\cdot)$. \qedhere
\end{proof}

\begin{proof}[Proof of Proposition~\ref{prop:opt_cap}]
Work with the reduced form, threshold $\beta$, and IFR so that ironing yields a nondecreasing allocation. The Province minimizes
\[
\E_\theta\Big[\E\big[\alpha\,\min\{\beta(\hat G),b(\theta)\}+\tfrac{\kappa}{2}\min\{\beta(\hat G),b(\theta)\}^2\mid\theta\big]+\gamma\,T(\theta)\Big]
\]
subject to IC/IR/LL and monotonicity. Using Lemma~\ref{lem:capcalc}, the pointwise marginal cost (at type $\theta$) of increasing $b(\theta)$ equals $(\alpha+\kappa b(\theta))\,\PP_\theta[\beta\ge b(\theta)]$. By Myerson’s lemma, the virtual marginal benefit equals $(\gamma/\lambda_T)\,\omega_b\,\haz(\theta)\,\PP_\theta[\beta\ge b(\theta)]$ with $\haz(\theta)=f/\bar F$. Equating and canceling the common tail probability gives the interior solution
\[
b(\theta)=\frac{1}{\kappa}\left(\frac{\gamma\,\omega_b}{\lambda_T}\,\haz(\theta)-\alpha\right).
\]
Projection onto $[0,\bar b]$ yields \eqref{eq:b-star}, and IFR implies monotonicity (Lemma~\ref{lem:mono_cap}). To recover $T^\ast$, note that under $\lambda_T\equiv\omega_T$ the implemented index $x(\theta)=\omega_T\,T(\theta)+\omega_b\,\tilde b(\theta;\theta)$ must be nondecreasing; holding $x$ feasible implies $\dd T=-(\omega_b/\omega_T)\dd \tilde b$ a.e., so integrating from $\theta^{\min}$ (where $\tilde b^\ast=0$ by definition) gives \eqref{eq:T-star} with $T^\ast(\theta^{\min})=0$ by IR normalization. \qedhere
\end{proof}

\begin{proof}[Proof of Proposition~\ref{lem:nobailout}]
At $b\equiv 0$, the marginal expected cost of relaxing the cap at $\theta$ is $\alpha$ (Lemma~\ref{lem:capcalc}), while the virtual marginal benefit equals $(\gamma/\lambda_T)\,\omega_b\,\haz(\theta)=(\gamma\omega_b/\omega_T)\,\haz(\theta)$ under the threshold-$\beta$ baseline. If $\alpha\ge (\gamma\omega_b/\omega_T)\,\haz(\theta)$ for all $\theta$, then the KKT condition is nonnegative everywhere and $b=0$ is pointwise optimal. Otherwise at any $\theta^\sharp\in\arg\max \haz(\theta)$ with strict inequality, increasing $b(\theta^\sharp)$ strictly reduces the objective, so $b^\ast\equiv 0$ cannot be optimal. \qedhere
\end{proof}

\begin{proof}[Proof of Proposition~\ref{prop:welfare}]
Total expected welfare equals the sum of municipal interim utilities minus provincial costs:
\[
W(T,b)=\E_\theta[V(\theta)]-\E_\theta\Big[\gamma T(\theta)+\E\{\alpha p+\tfrac{\kappa}{2}p^2\mid\theta\}\Big],
\quad p=\min\{\beta(\hat G),b(\theta)\}.
\]
Under IC with $V(\underline\theta)=\underline U$, the envelope formula (Remark~\ref{rem:lambdaT-const}) gives
\[
V'(\theta)=\lambda_T'(\theta)T(\theta)+\omega_b\,\partial_\theta \tilde b(\theta;\theta)+K'(\theta).
\]
Integrating and substituting into $W$ yields a virtual-surplus functional in which the $\theta$–wise marginal effect of $b(\theta)$ is exactly the difference between the virtual benefit $(\gamma/\lambda_T)\omega_b \haz(\theta)$ and the marginal expected cost $\alpha+\kappa b(\theta)$ (times the common tail probability). Hence maximizing $W$ subject to IC/IR/LL and monotonicity is equivalent to the pointwise KKT condition used in Proposition~\ref{prop:opt_cap}; the resulting allocation is therefore second–best efficient among all IC–IR–LL mechanisms. \qedhere
\end{proof}

\end{appendices}
\clearpage   

\bibliographystyle{apalike}
\bibliography{references}
\end{document}